\documentclass[fleqn,10pt]{wlscirep}

\usepackage[utf8]{inputenc}
\usepackage[T1]{fontenc}
\usepackage{amsmath, amsthm, amssymb, amsfonts}  
\theoremstyle{plain}
\newtheorem{theorem}{Theorem}[section]

\usepackage[english]{babel}
\usepackage{graphicx}

\usepackage{setspace} 
\def\be{\begin{equation}}
\def\ee{\end{equation}}
\def\bea{\begin{eqnarray}}
\def\eea{\end{eqnarray}}

\def\e{\varepsilon}
\def\epsilon{\e}

\def\t{\tau}

\newcommand{\nocontentsline}[3]{}
\newcommand{\tocless}[2]{\bgroup\let\addcontentsline=\nocontentsline#1{#2}\egroup}

\DeclareMathSymbol{\leqslant}{\mathalpha}{AMSa}{"36} 
\DeclareMathSymbol{\geqslant}{\mathalpha}{AMSa}{"3E} 
\DeclareMathSymbol{\eset}{\mathalpha}{AMSb}{"3F}     
\renewcommand{\leq}{\;\leqslant\;}                   
\renewcommand{\geq}{\;\geqslant\;}                   

\title{Contingent Convertible Bonds in Financial Networks}

\author[1]{Giovanni Calice}
\author[2]{Carlo Sala}
\author[3,*]{Daniele Tantari}

\affil[1]{School of Business and Economics, Loughborough University, Loughborough, LE11 3TU, United Kingdom}
\affil[2]{Department of Economics, Finance and Accounting, ESADE Business School (Ramon Llull University), Avenida de Torreblanca 59, Barcelona, Spain}
\affil[3]{Mathematics Department, University of Bologna, Piazza di Porta San Donato 5, 40126, Bologna, Italy.}
\affil[*]{daniele.tantari@unibo.it}

\begin{abstract}
We study the role of contingent convertible bonds (CoCos) in a complex network of interconnected banks. 
By studying the system's phase transitions, we reveal that the structure of the interbank network is of fundamental importance for the effectiveness of CoCos as a financial stability enhancing mechanism. 
Our results show that, under some network structures, the presence of CoCos can increase (and not reduce) financial fragility, because of the occurring of unneeded triggers and consequential suboptimal conversions that damage CoCos investors. 
We also demonstrate that, in the presence of a moderate financial shock, lightly interconnected financial networks are more robust than highly interconnected networks. This makes them a potentially optimal choice for both CoCos issuers and buyers.
\end{abstract}

\begin{document}

\flushbottom
\maketitle

\thispagestyle{empty}

\section*{Introduction}
\label{sec: Introduction}

The 2007-2009 financial crisis highlighted the critical role of interbank interconnectedness in the stability of the global financial system and underscored a network-based approach to develop effective strategies for mitigating financial risk \cite{battiston2012debtrank,cimini2015systemic,somin2020network,petrone2018dynamic,boersma2020reducing,so2022assessing}.
Designed to reduce the impact of a lack of short-term liquidity in times of financial distress, Contingent Convertible bonds (henceforth, CoCos) have been extensively issued in the aftermath of the 2008-2009 financial crisis, with the goal of serving as a protective buffer during adverse times.
CoCos are coupon-paying bonds that, either convert into equity shares, or are (fully or partially) written-off, when the issuer reaches a pre-specified level of financial distress.
Hence, CoCos serve as regulatory instruments designed to absorb unexpected future losses of the issuing bank through automatic recapitalization triggered at a predefined level. 
This mechanism provides additional loss-absorbing capital to undercapitalized banks during periods when raising fresh equity capital would be challenging.

First proposed by \cite{Merton_1991}, who initially described the use of CoCos as a capitalization buffer during economic downturns, the literature on CoCos began to take shape with the work of \cite{Flannery_2002} and \cite{Flannery_2005}. 
As documented in a study by \cite{AVDJIEV2020}, since the first issuance by Lloyd's, over 500 billion US dollars' worth of CoCos have been issued, spanning more than 400 different issues by various banks across different countries.
In \cite{Fajardo_2019} is documented that the majority of banks issuing CoCos are large, highly leveraged institutions, particularly in BRICS and emerging economies, with the primary goal of meeting the risk provisions mandated by Basel III requirements.
A comprehensive survey of this literature can be found in \cite{Flannery_2014,AVDJIEV2020}, along with their respective references.
Notice that an effective modeling and design of CoCos is still a fundamental open and unresolved research question in the academic literature. Models where CoCos have a conversion trigger linked to accounting values are developed in \cite{Calomiris_2013,Bolton_2012,McDonald_2013,Pennacchi_2014}, while CoCos configuration from a market value perspective are analyzed in \cite{Sundaresan_2015,Glasserman_2016,Pennacchi_2018,Pennacchi_2019}.
Finally, the incentive effects of CoCos in individual banks have been investigated in \cite{Hori17}.
Overall, it is worth noticing that, due to their complexity, the configuration of CoCos is not trivial and, both academics and practitioners, still disagree on how CoCos should be structured and priced (see \cite{Flannery_2014,Greene_16,Oster_2019}).
The role of CoCos in bank capital regulation is explored in \cite{Kashyap_2008} and \cite{Squam_2010}.
Finally, \cite{Koziol_2012,Hilscher14,Berg11,Chan_2017,Goncharenko_2017,Martynova_2018,Albul_2013,Chen_2017_RFS,Goncharenko_2019} discusses how the possible configurations of CoCos can lead to corporate governance problems like debt overhang, or to risk shifting/taking incentives and possible bank failures due to extreme deleveraging.

To the best of our knowledge, this is the first paper, along with \cite{Gupta2020}, to analyze the role of CoCos in an interbank network.
Differently from \cite{Gupta2020} we focus more on the interbank network to study the contagion and stability effects of CoCos.
Specifically, we study whether the introduction of CoCos can effectively reduce the overall amount of systemic risk in the economy, in particular by mitigating the extent of default propagation that may result in systemic failures.
Additionally we investigate whether, also in presence of CoCos, interconnected systems exhibit the 'robust-yet-fragile' phenomenon discussed by (\cite{Haldane_14,Acemoglou15,Gai_2010,brini2023reinforcement}).

To achieve this, we expand upon the work of~\cite{Acemoglou15} by incorporating CoCos into the financial system and exploring their interactions with the network's topology, thereby influencing financial stability.
To shed light on the role of CoCos as financial stability instruments, we first revisit in Sections~\ref{sec: The interbank model} and~\ref{sec:Financial Networks with no CoCos}) the main results concerning the propagation of liquidity distress caused by exogenous shocks in the absence of CoCos.
Next, in Section \ref{sec:Financial Network with CoCos}, we introduce CoCos in the banks' balance sheets, and analyze their effects within the financial network.
We initially focus solely on CoCos, abstracting from the effects of equity conversion in the connected financial network. Subsequently, in Section~\ref{Subsec: Cocos with equity liquidation}, we consider a scenario where all banks in the financial network can effectively utilize CoCos proceeds.
Our main findings confirm the robust-yet-fragile nature of financial networks, as documented in the absence of CoCos (\cite{Allen_2000},~\cite{Freixas_2000}).
This characteristic persists in the presence of CoCos, where greater robustness to the contagion's trigger corresponds to a more pronounced impact of the contagion once it begins.
In particular lightly interconnected networks, that are more prone to contagion also for small shocks, demonstrate greater robustness to moderate shocks in terms of the extent of contagion when compared to highly interconnected networks.
Finally, we find that the effectiveness of CoCos as a financial stability-enhancing mechanism is contingent upon the type of network and may not always benefit the economy.
In fact, in the presence of moderate shocks, fully and highly interconnected networks can act as sources of unnecessary triggers, which are solely a consequence of network interconnections rather than bank defaults. 
This can lead to suboptimal conversions, potentially harming CoCos investors who may no longer receive coupon payments.


\section{The interbank model}
\label{sec: The interbank model}

The inclusion of any debt securities on a bank's balance sheet establishes bilateral obligations between the issuer and the owner, which can be represented by constructing an interbank network with $n\in\mathbb{N}$ nodes. 
In the network, each bank is represented by a node, and the obligations are denoted by directed edges connecting these nodes.
Specifically, a directed edge from node $i$ to node $j$ exists if bank $i$ is creditor of bank $j$, such that $y_{ij}$ represents the face value of the contract among the two banks.


Following~\cite{Acemoglou15}, we focus on three time periods, denoted as $t_0$, $t_1$ and $t_2$.
At $t_0$ each bank $i$ is endowed with an initial capital.
This initial capital can be lent to other banks, held as cash $c_i$, or invested in competitive projects.
If lent to other banks, the interbank lending takes place at $t_0$, and banks use the money borrowed to finance their investments.
The investment can yield two types of returns: i) a short-term $t_1$ stochastic return $r_i$, or ii) a long-term $t_2$ deterministic return $A_i$, if the project is held until maturity.

At time $t_1$ banks honor their senior and interbank obligations.
Senior obligations (e.g. taxes, wages), denoted as $s_i>0$, are non-negative external liabilities that have the highest priority in repayment.
For the interbank obligations, banks pay an interest rate on the principal, so that the face vale $y_{ij}$ of the $j$ debt to bank $i$ is the product of the amount borrowed and the interest rates.
Identifying with $y_i=\sum_{j\neq i} y_{ji}$ the bank's $i$ interbank obligations, it follows that the bank $i$ total liabilities at time $t_1$ are $\sum_{j \ne i} y_{ji} + s_i = y_i + s_i$.
In terms of liabilities liquidation, junior (interbank) debts have all equal seniority, and are paid after the senior debts.
After senior debts are settled, in the event of a company default, junior debts are repaid proportionally based on their face values. Denoting what bank $i$ returns to bank $j$ as $x_{ji}$, it holds
$$
x_{ji}=\phi_i \ y_{ji}, \quad \phi_i \in [0,1] \quad \text{for any} \, \, i,j=1,\ldots, n.
$$
The parameter $\phi_i$ indicates the bank's fitness, representing the fraction of the junior obligations that bank $i$ can repay to its creditors. If $\phi_i < 1$ then bank $i$ cannot honor its junior debt, while if senior debts cannot be honored, all junior debts remain completely unpaid, i.e., $\phi_i = 0$.

Banks can honor their obligations by using all the resources available at time $t_1$. In a simplified model they include the liquidity (or liquid assets) available $z_i=c_i+r_i$, i.e. the sum of cash $c_i$ and short term return from the project $r_i$, and what debtor banks effectively return at time $t_1$, i.e. $\sum_{k\neq i} x_{ik}$. 
If necessary, banks have the option to fully or partially liquidate their long-term investments, denoted as $A_i$, to acquire an additional liquidity budget of $\zeta l_i$, where $l_i \in [0, A_i]$ represents the bank's liquidation decision, and $\zeta \in [0, 1]$ signifies the cost of liquidating before maturity.

From this mechanism the repayment structure depends on the maximum amount of resources available at time $t_1$, i.e. $h_i:=z_i+\zeta A_i +\sum_{k\neq i} x_{ik}$ in relation to the total obligations, i.e. $s_i+y_i$. 
It follows that bank $i$ returns to bank $j$ the amount
\begin{equation} \label{eq:debtrule}
    x_{ji} =\phi_i \ y_{ji}=
\begin{cases}
    y_{ji},         & \text{if} \quad h_i > s_i + y_i \\
    (h_i-s_i)y_{ji}/y_i & \text{if} \quad h_i \in (s_i, s_i+y_i) \\
    0               & \text{if} \quad h_i \in (0, s_i)
\end{cases}
\end{equation}
where, in the intermediate case of junior insolvency, bank $i$ repays her creditors a fraction $y_{ji}/y_i$ of the remaining resources available after senior debts are settled. 
In the other two cases bank $i$ is either fully solvent $x_{ji}=y_{ji}$ ($\phi_i=1$), or fully insolvent $x_{ji}=\phi_i=0$.  
Note that the liquidation decision $l_i$ becomes irrelevant because when a bank is insolvent, a full liquidation of the project is required, i.e., $l_i = A$.
Eq. ($\ref{eq:debtrule}$) can be written in a more compact form as 
\begin{equation}
x_{ji}=\frac{y_{ji}}{y_i} \min ( y_{i},h_i - s_i) ^{+}=y_{ji} \min \left( 1,\frac{h_i - s_i}{y_i}\right) ^{+},
\label{eq:debt repayment rule}
\end{equation}
where $[ \cdot ]^+ = \max[\cdot, 0]$.
Dividing Eq. ($\ref{eq:debt repayment rule}$) by $y_{ji}$ and using $\phi_i = x_{ji}/y_{ji}$, we can derive the financial distress propagation rule in terms of the bank's fitness:
\be\label{rule}
\phi_i =  \min\left(1, \frac{h_i(\phi)-s_i}{y_i} \right)^+ = f_{y_i,s_i}(h_i(\boldsymbol{\phi})),
\ee
Here, we emphasize that the available resources $h_i = z_i + \zeta A_i + \sum_{k\neq i}y_{ik}\phi_k$ depend on the fitness of other banks. We also define the activation function as $f_{y,s}(h) = \min(1, \frac{h - s}{y})^+$.
Eq. $\eqref{rule}$ can be thought either as an updating rule for fitness propagation, $\phi^{t+1}_i= f_{y,s} (h_i(\boldsymbol{\phi}^t))$, or as an equilibrium defining map $F: [0,1]^n \to [0,1]^n$, $ \boldsymbol{\phi}=(\phi_1,\ldots,\phi_n) \to F(\boldsymbol{\phi})=\left( f_{y,s}(h_i(\boldsymbol{\phi}))\right)_{i=1}^n$.

In the remainder of the paper, we focus on a system of homogeneous banks, where we assume that $s_i = s$ and $y_i = y$ for all $i = 1, \ldots, n$. Furthermore, Eq. \eqref{eq:debt repayment rule} suggests that we can make the assumption, without loss of generality, that $A_i = 0$, as project liquidation simply acts as an additional asset and can thus be absorbed into the variable $z_i$.
Lastly, we consider the possibility of banks experiencing exogenous liquidity shocks. 
For each bank $i$, we model this situation as $z_i = a > s$ in the absence of shocks and as $z_i = a - \varepsilon$ in the presence of shocks.

In the above setting, starting from the initial condition $\phi^0 = (1, \ldots, 1)$ we iterate Eq. ($\ref{rule}$) until convergence to an equilibrium fitness which is the unique\cite{Acemoglou15} fixed point.
We therefore investigate the equilibrium's properties in terms of two key factors: 1) the extent of contagion denoted as $E(\boldsymbol{\phi})$ and 2) the system's distress represented as $D(\boldsymbol{\phi})$. These are defined as follows:
\begin{equation}
E(\boldsymbol{\phi}) = 1 - \frac{1}{n} \sum_{i=1}^n \delta_{\phi_i,1}; \ \ \ \ D(\boldsymbol{\phi}) = 1 - \frac{1}{n} \sum_{i=1}^n \phi_i.
\label{eq:extent_contagion}
\end{equation}
Additionally, we study the financial stability of the system as a function of 1) the \textit{topological properties} of the interbank directed and weighted network $Y=(y_{ij})$; and 2) the \textit{size} and \textit{distribution} of the shocks.

\section{Financial Networks without CoCos}
\label{sec:Financial Networks with no CoCos}

The focus of the present analysis is mainly on the role of the connectivity as a first order possible metrics of network topology.  To isolate the effect of the network connectivity from the possible noise due to node heterogeneity we consider regular networks. In the Supplementary Material we also consider more realistic networks with some level of degree heterogeneity and a macroscopic structure (e.g. assortative communities or core-periphery structure), by showing that the non trivial behavior of the contagion still can be interpreted to some extent as a result of the network local connectivity. 
A network is regular if $\sum_{j\neq i}y_{ij}=\sum_{j\neq i}y_{ji}=y$, i.e. everyone owes everyone the same amount. 
Two particular extreme cases are represented by the ring and the complete regular networks.
Specifically, a financial network is a ring network if $y_{i,i-1}=y_{1,n}=y$, and $y_{ij}= 0$ otherwise.
Under this configuration, bank $i$ is the unique creditor of bank $i-1$, and bank $1$ is the unique creditor of bank $n$, so that a default of a bank spillovers \textit{entirely} on the subsequent banks. 
Conversely, a financial network is a complete network if $y_{ij}=\frac{y}{n-1}\forall i\neq j$.
Under this setting a liability, and thus a possible bank default, is equally divided among all $n$ banks in the financial network.
With the aim of studying more realistic networks that better resemble an interbank network, we also introduce and focus most of our attention on random regular networks with different degrees of connectivity, $0 < c < \infty$. 
Connectivity is defined as the number of incoming (or outgoing) links in the network, represented as $c=\sum_{j\neq i} \mathbb{I}_{y_{ij}>0}=\sum_{j\neq i} \mathbb{I}_{y_{ji}>0}$.
Assuming the same connectivity for each bank, we obtain a regular network by dividing all the junior liabilities $y$ of a bank in equal parts between its neighbors, i.e. $y_{ij}= \frac{y}{c} \mathbb{I}_{y_{ij}>0}.$

The principal discovery in the study by \cite{Acemoglou15} concerning the resilience of ring and complete financial networks under the influence of a non-negative exogenous shock $\epsilon>0$ can be summarized (please refer to the Supplementary Material) in the following
\begin{theorem}
\label{th:nococo}
Given $\varepsilon^{\ast} = n\left( a-s\right)$ and $y^{\ast} = (n-1)\left( a-s\right)$, then:
\begin{itemize}
\item as soon as $\varepsilon < \varepsilon^{\ast}$ (small shock regime) or $y < y^{\ast}$ (low exposure regime) the extent of contagion in the ring network is larger than that in the complete network.
\item as soon as $\varepsilon > \varepsilon^{\ast}$ and $y > y^{\ast}$ default becomes systemic in both the ring and the complete networks.
\end{itemize}
\end{theorem}
Fig. \ref{fig:RR} illustrates the extent of financial contagion and banking distress as functions of the shock in the high exposure regime, ($y>y^\star$), with ring and complete networks represented by blue and green lines, respectively.
\begin{figure}
\includegraphics[width=0.5\linewidth]{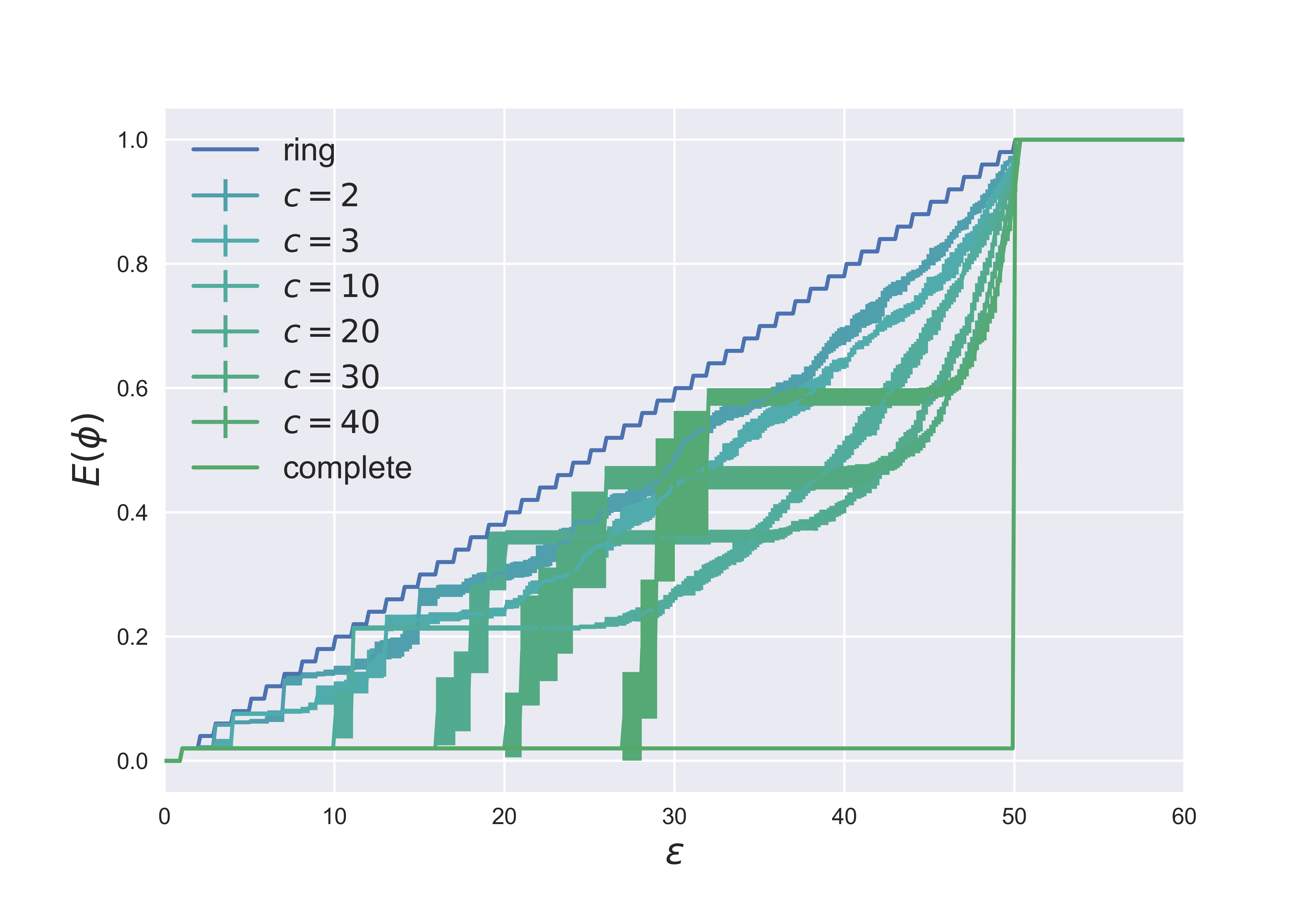}
\includegraphics[width=0.5\linewidth]{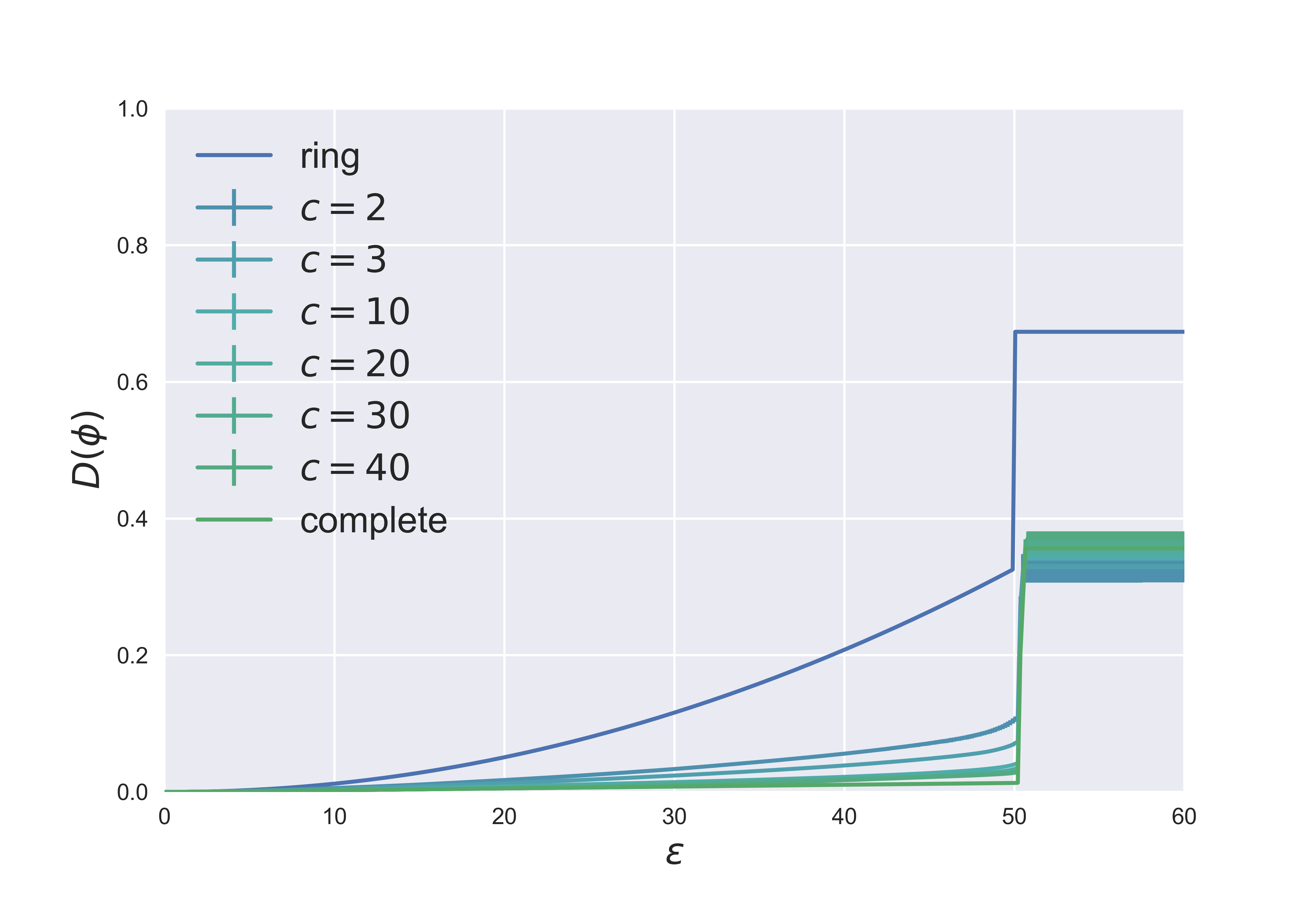}
\caption{Simulation on Random Regular networks with $N=50$, $a=21$, $s=20$, $y=75>y^\star$.
Results are averaged over $10$ different realizations, sampled from a directed configuration model.
Left panel: extent of contagion; Right panel: system's distress.
\label{fig:RR}}
\end{figure}
Two important messages emerge from the figure.
First, it is evident how the contagion becomes systemic at the same point $\epsilon = \epsilon^*$.
Secondly, contagion (default leading to default) is quicker with the ring network for small shocks. 
On the other hand, large shocks, which by definition envisage less interconnection among banks, lead to senior debts being wiped out to absorb losses.
In this regime, the ring network becomes as stable as the complete network, thus reflecting a robust-yet-fragile framework.

Fig. \ref{fig:RR} also presents, for different intermediate network connectivity, the average results over 10 realizations of random regular networks sampled from a directed configuration model \cite{bollobas1980probabilistic, Newman}.
Interestingly, when the extent of financial contagion is nearly zero, the results confirm that the complete network exhibits greater stability to small shocks, compared to all other networks.
In fact, for the ring and all intermediate networks, the extent of contagion increases with the shock size. 
In contrast, the complete network only experiences systemic contagion for $\epsilon>\epsilon^\star$, marking the previously defined phase transition with a jump.
Focusing on the intermediate networks, we can delve deeper into their behavior and responses to shocks. Our findings reveal the existence of three regimes, not just two. 
Specifically, as we progress from the ring network and increase the degree of connectivity, the system becomes progressively more resilient, with a notable jump in the extent of contagion as the shock size increases.
The jump corresponds to the point at which the shock propagates simultaneously to all neighbors of the stressed bank. 
In the limit, as the network becomes fully connected, the system remains stable until it reaches the transition point where the shock becomes systemic.
It is interesting to note that there are shock size regimes where increasing connectivity doesn't necessarily result in a greater stability. 
A more connected network tends to be more stable only \textit{before} the first jump. 
Beyond that point, the system becomes more fragile compared to cases with lower connectivity.


\section{Financial Network with CoCos}
\label{sec:Financial Network with CoCos}

The situation is different when we consider all bilateral obligations between banks to be CoCo bonds. 
A crucial element in any CoCo configuration is the trigger, which defines a critical level of the issuer's financial distress. 
When the trigger condition is met, the bond is either partially or fully written off. 
As a result, the issuer bank experiences a reduction in its exposure while the owner bank receives some equity shares in exchange (for CoCos with equity-conversion).
Various types of trigger mechanisms exist~\cite{Flannery_2014}, depending on the definition of financial distress (e.g., market- or book-value trigger). 
In our model, we define it as the ratio between a proxy of the bank's $i$ equity, $E_{i}$, and the bank's assets, $h_{i}$, at time $t_1$. 
Recall that $h_i = z_i + \sum_{j\neq i} x_{ij}$ represents the total amount of resources available, and that $s+y$ denotes the total obligations of bank's $i$ at time $t_1$. 
Under this framework, we assume a trigger to occur if
\begin{equation}
\label{eq:tr2}
\frac{E_{i}}{h_{i}}:= \frac{h_i-(s+y)}{h_i}\leq \tau \ \ \iff \ \ h_{i}\leq \frac{s+y_{i}}{1-\tau},
\end{equation}
where $\tau>0$ is a discretionary threshold.

In the case of CoCos with equity conversion, when a trigger event occurs, the bonds are either partially or fully converted into equity to address, if feasible, the shortfall up to the threshold $E_i = \tau h_i$.
Suppose that, before conversion, $E_i = \beta h_i$, with $\beta\leq \tau$.
Under this scenario the amount of bond $\Delta y$ that has to be converted into equity is $\Delta y = (\tau - \beta)h_i$, provided that $y>\Delta y$.

In the case $y \leq \Delta y$, then the entire CoCo is converted: this happens if
\begin{eqnarray}\label{eq:tr}
y \leq \Delta y \iff h_i - s - \beta h_i \leq (\tau-\beta)h_i \iff h_i \leq \frac{s}{1-\tau}.
\end{eqnarray}
From Eqs. ($\ref{eq:tr2},\ref{eq:tr}$) we can define three different scenarios.
\begin{itemize}
\item If $\frac{s}{1-\tau} \leq h_i\leq \frac{s+y}{1-\tau}$, the bond undergoes partial conversion into equity: in this case, CoCos holders receive $\Delta y = (\tau - \beta)h_i$ of converted equity, and $x_i=\sum_{j\neq i}x_{ji} = y - \Delta y = (1-\tau)h_i -s$ of unconverted CoCos.
\item If $s \leq h_i\leq \frac{s}{1-\tau}$, the entire bond is converted: in this case, CoCos holders hold $\Delta y = y $ of converted equity and no unconverted CoCo, i.e. $x_i=0$.
\item If $h_i\leq s$, equity holders absorb the losses ahead of senior creditors, and consequently, they do not receive converted equity. 
In this case, $x_i=0$.
\end{itemize}
When we combine all these cases together, we obtain
$x_i = \min [y, (1-\tau)h_i -s]^+$.
Then, dividing by $y$, we can obtain the equivalent trigger propagation rule but in terms of the banks fitness $\phi_i =x_{ij}/y_{ij}= x_i/y_i$ as
\be\label{rulecoco}
\phi_i  = f_{y,s}((1-\tau) h_i(\boldsymbol{\phi})),
\ee
where $f_{y,s}$ is defined in Eq. ($\ref{rule}$). The main difference between Eq.~\eqref{rulecoco} and Eq.~\eqref{rule} lies in the interpretation of the fitness parameter $\phi_i$. 
In the presence of CoCos, $\phi_i < 1$ does not signify default but rather the triggering of the CoCo. 
Consequently, distress propagation in this system indicates a triggering propagation. 

It is worth noticing that the continuous map defined in Eq. ~\eqref{rulecoco} admits a unique fixed point, regardless of the network structure (see Supplementary Material). 
As in the previous section, after an exogenous liquidity shock, we can analytically derive the unique equilibrium for both the ring and complete networks. 
The properties of this equilibrium are summarized in the following 
\begin{theorem}\label{th:coco}
There exist exposure thresholds 
$y^{\ast}_r(\tau)$, 
$y^{\ast}_c(\tau)$ 
and shock thresholds $\varepsilon^{\ast}_r(y,\tau)$, $\varepsilon^{\ast}_c(y,\tau)$ (their explicit expressions are provided in the Supplementary Material) such that:
\begin{itemize}
\item when $(y,\tau)\in \mathcal{S}_r =\{\varepsilon > \varepsilon^{\ast}_r(y,\tau), y > y^{\ast}_r(\tau)\}$ the shock triggers a systemic CoCo triggering in the ring network.
\item when $(y,\tau)\in \mathcal{S}_c =\{\varepsilon > \varepsilon^{\ast}_c(y,\tau), y > y^{\ast}_c(\tau)\}$ the shock triggers a systemic CoCo triggering in the complete network.
\end{itemize}
where $\mathcal{S}_r$ and $\mathcal{S}_c$ identify the systemic \textit{unstable} regions for the ring and complete networks, respectively.
It follows that their complement, $\mathcal{S}^c_r$ and $\mathcal{S}^c_c$, identify the \textit{safe} regions for the ring and complete networks, respectively.
Moreover $\mathcal{S}_r\subset \mathcal{S}_c$ and
\begin{itemize}
\item when $(y,\tau)\not\in \mathcal{S}_c $(implying $(y,\tau)\not\in \mathcal{S}_r$) the ring network is the least (and the complete network is the most stable) financial network.
\item there exist a region where $(y,\tau)\not\in \mathcal{S}_r$ but $(y,\tau)\in \mathcal{S}_c$ where the ring network is the most (and the complete network is the least) stable financial network.
\end{itemize}
\end{theorem}

The proof, detailed in the Supplementary Material, relies once more on the analytical solution of the fixed-point Eq.~\eqref{rulecoco}.

The first part of the theorem asserts that, in the presence of CoCos, the unstable region (characterized by large shocks and high exposure) is not universal, as in the case of vanilla bonds. Instead, it strongly depends on the network structure, resulting in different thresholds for various levels of shock size ($\epsilon$) and exposure ($y$).
Additionally, as depicted in Fig~\ref{fig:phase}, the unstable regions (or equivalently the safe regions) are strongly influenced by the triggering parameter $\tau$.
\begin{figure}[ht!]
\includegraphics[width=0.5\linewidth]{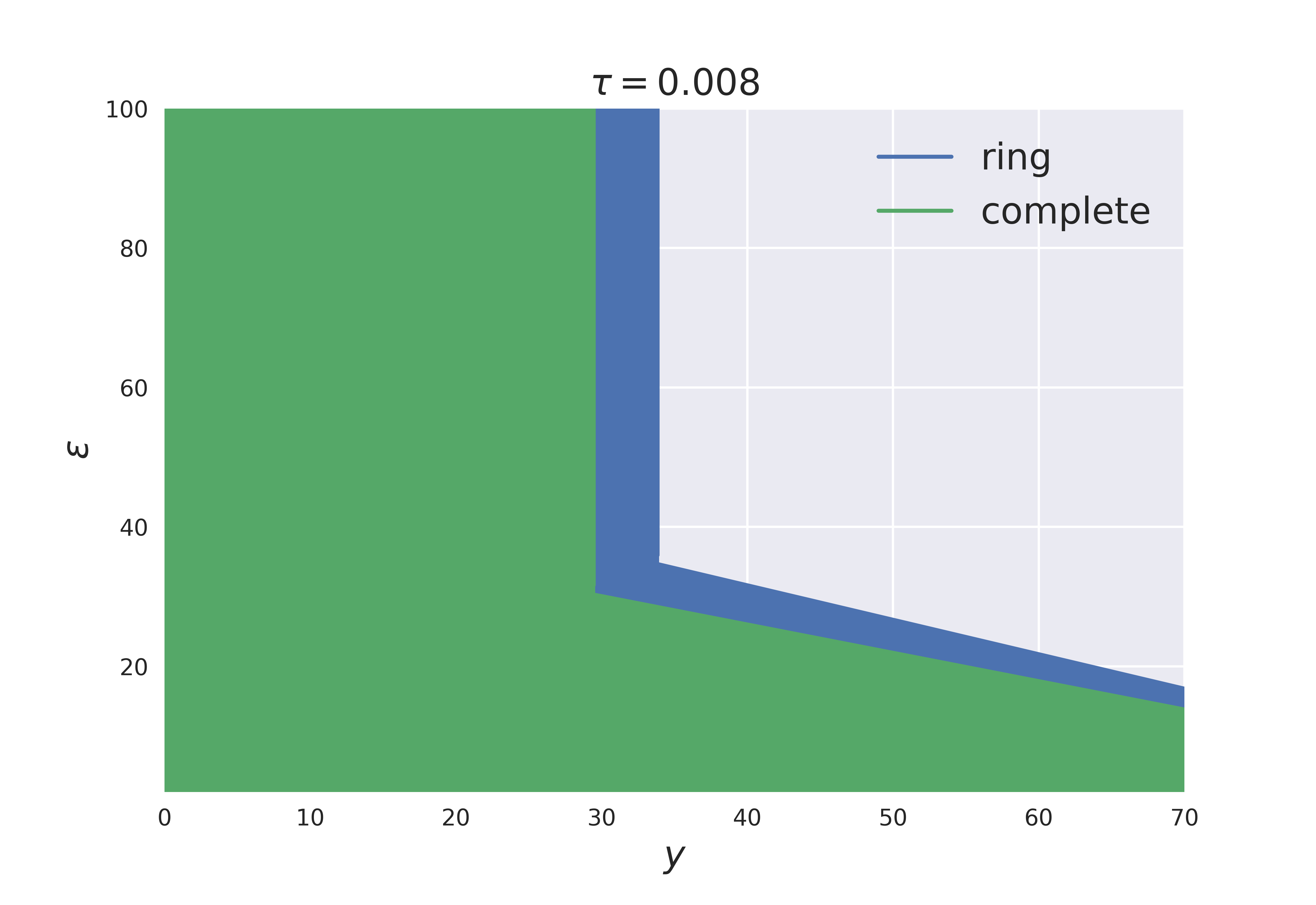}
\includegraphics[width=0.5\linewidth]{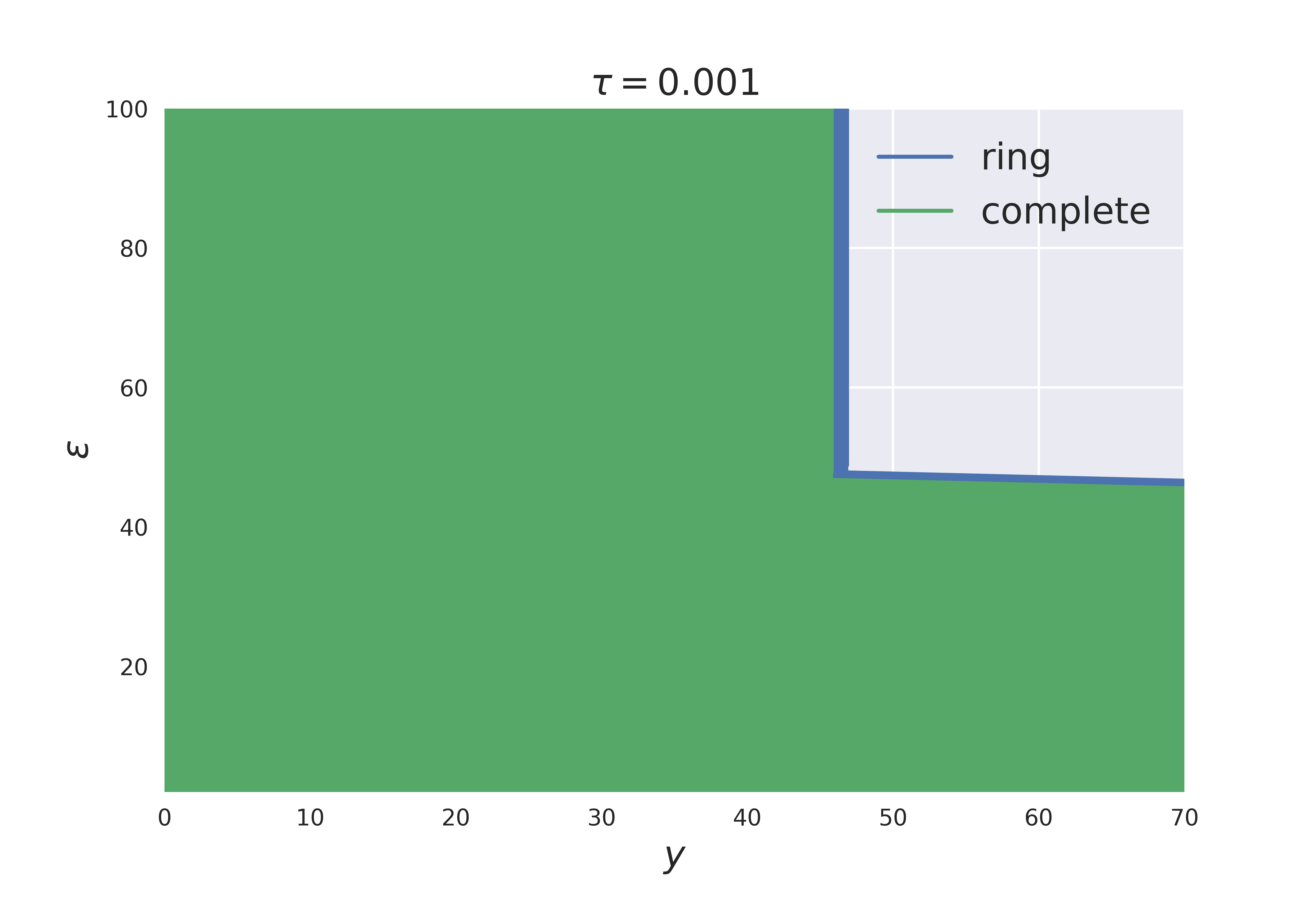}
\caption{Safe region for the ring and the complete networks for different values of $\tau$.
\label{fig:phase}}
\end{figure}
Naturally, as $\tau$ tends to zero, CoCos become vanilla bonds, and the two regions tend to coincide.
Conversely, as $\tau$ increases, the two regions tend to diverge more and more significantly.
In the second part of the theorem, the two unstable regions (for the ring and complete networks) are compared.
Remarkably, it appears that the safe region for the complete network is the smallest one.
This suggests that if the size of the shock is such that it does not determine a systemic triggering in the complete network, the latter is the most stable topology (first point).
However (second point), the triggering in the complete network becomes systemic for a shock size that is smaller than the one necessary for a systemic triggering in the ring network.
As such, and as illustrated by the blue region in Fig.~\ref{fig:phase}, there exists a medium shock size region where the ring network is more stable than the complete network.
For larger shocks, the triggering becomes systemic also in the ring network, and the two topologies are equivalently sub-optimal.

Fig.~\ref{fig:CORR} shows the extent of contagion and system's distress as a function of the shock size in a high exposure regime, for different network connectivity.
\begin{figure}[!htb]
\includegraphics[width=0.5\linewidth]{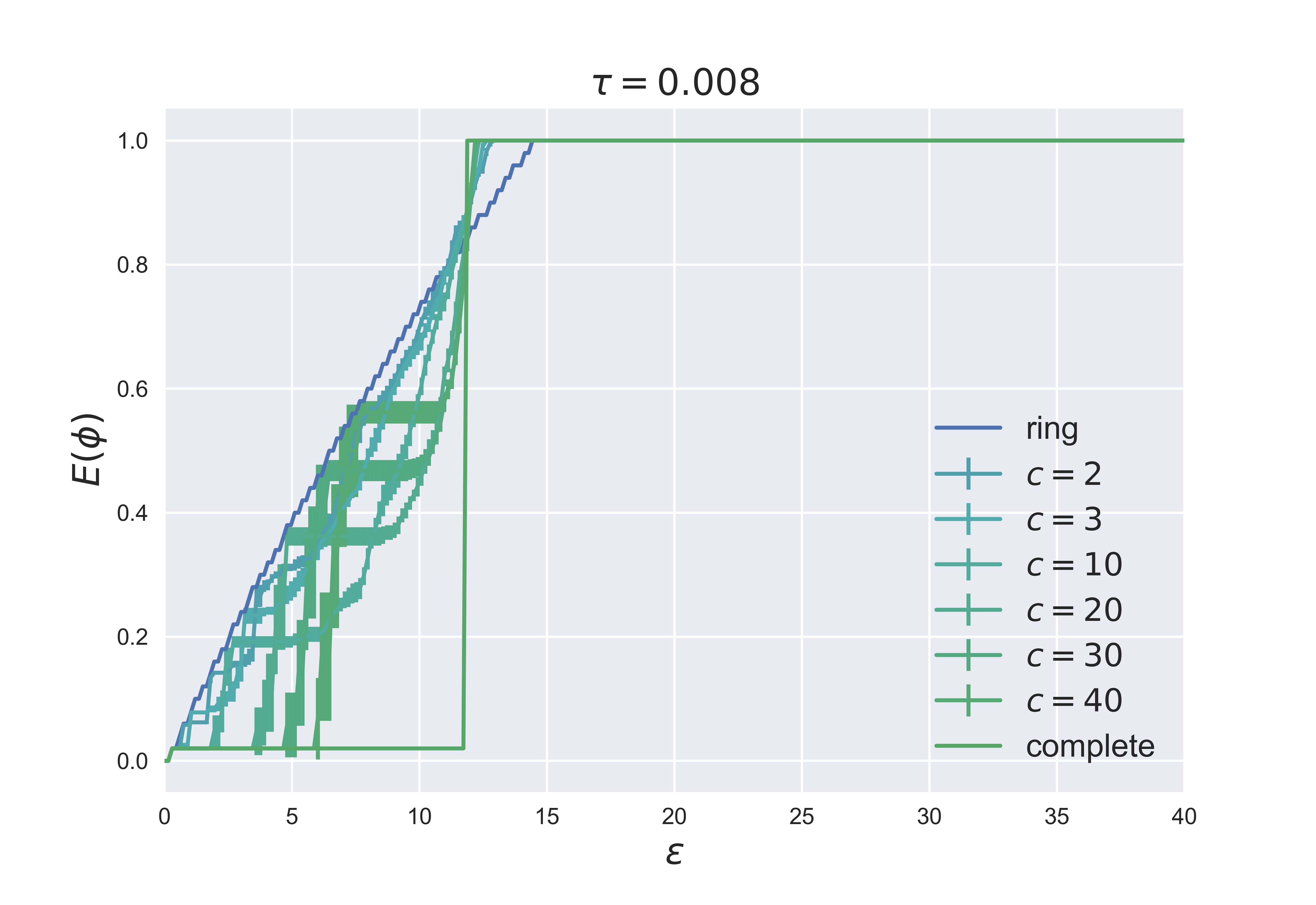}
\includegraphics[width=0.5\linewidth]{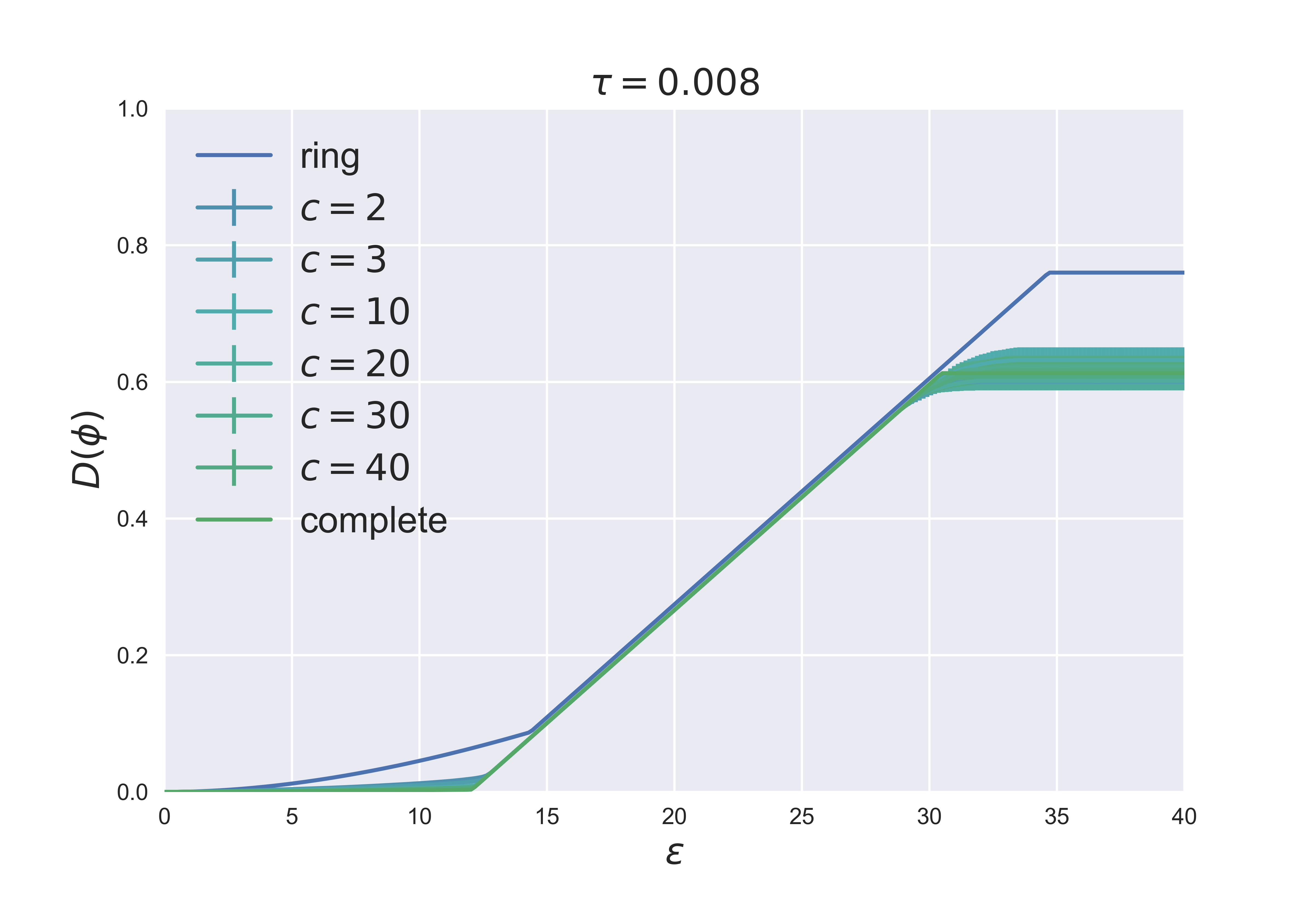}
\caption{Simulation on random regular networks with $N=50$, $a=21$, $s=20$, $y=75$.
Results are averaged over $10$ different realizations sampled from a directed configuration model.
\label{fig:CORR}}
\end{figure}
From the figure, we can observe that the system behaves similarly to the vanilla bond case for relatively small shocks. 
The ring network exhibits the highest resilience, the complete network the lowest, and intermediate connectivity levels result in jumps that define the region where higher connectivity enhances the system's robustness.
Note that, immediately after a jump occurs, the network enters a fragility phase where systems with lower connectivity are more robust.
For medium shock sizes, an inversion point is reached where all lines in Fig.~\ref{fig:CORR} intersect. 
Beyond this point, contagion becomes systemic in the complete network while the ring network remains relatively stable.
Finally, for larger shocks sizes, all the network structures are equivalent, since they cannot prevent the triggering to become systemic.
A final remark on the difference between CoCos and vanilla bonds concerns the shape of the shock size threshold $\epsilon(y,\tau)$ which now explicitly depends on the bank exposure $y$ (instead of $\tau$).
In particular, for both network topologies, the critical threshold decreases with the exposure (the higher the exposure the smaller the safe regions), while in the presence of vanilla bond ($\tau=0$) it is independent.
An important financial stability consequence is the existence of a maximum exposure, beyond which any (even small) shock can trigger a distressed bank and potentially lead to a systemic crisis.

\subsection{CoCos with equity liquidation}
\label{Subsec: Cocos with equity liquidation}

As a final step, we analyze the scenario of a trigger event in which the CoCos holder can liquidate (at a liquidation cost) the dollar amount of shares of the CoCos received as a consequence of the trigger.
Note that, neglecting this monetary value, would result in an overestimation of both the shock size and the degree of systemic distress, or, equivalently, an underestimation of the systemic crisis threshold.
From the previous section, if a shock in the economy leads to a trigger, then the CoCos issuer bank $i$ repays an amount $x^{nc}_i$ of \textit{unconverted} CoCos
$x^{nc}_i = \min [y, (1-\tau)h_i -s]^+$
while it \textit{converts} and uses the remaining amount $x^c_i=y-x^{nc}_i$.
If we denote with $\eta\in[0,1]$ the effective market value of the issuer bank's share after the trigger, we can generalize the debt repaying rule as
$x_i = x^{nc}_i+ \eta x_i^c= \eta y_i + (1-\eta) \min [y, (1-\tau)h_i -s]^+$
and, equivalently, the propagation rule in terms of the bank fitness becomes:
\begin{equation}
\phi_i =  \eta + (1-\eta)f_{y,s}( (1- \tau)h_i(\boldsymbol{\phi})),
\end{equation}
that generalizes Eq. ($\ref{rulecoco}$) and reduces to it in the limit of $\eta\to 0$. 
Again, there is an unique equilibrium regardless of the network structure (see Supplementary Material), that can be reached by iteration, starting from $\phi=(1,\ldots ,1)$ and propagating  an exogenous liquidity shock. 
The value of $\eta$ depends on various internal and external factors, including the overall quality of the bank, the state of the banking system, and the size of the shock affecting an individual bank.
Moreover, a CoCo conversion might depress the bank's equity, being a trigger usually seen as a bad signal for the bank issuer by market participants.
Acknowledging all these factors and their consequences, in this paper, we treat $\eta$ as exogenously given, and analyze the effects on the level of bank fitness for different values of $\eta$.\\
First, for $\eta=0$ (shares have zero value) the bank cannot use any money from the conversion and we indeed retrieve the results of the previous section.
Second, for $0 < \eta \leq 1$ the equity conversion mitigates the propagation of the shock in a non-linear way, and $\eta$ becomes the minimum possible bank fitness.
Fig.~\ref{fig:CORRd} shows the extent of contagion as a function of the shock for different values of $\eta$. 
\begin{figure}[!htb]
\includegraphics[width=0.5\linewidth]{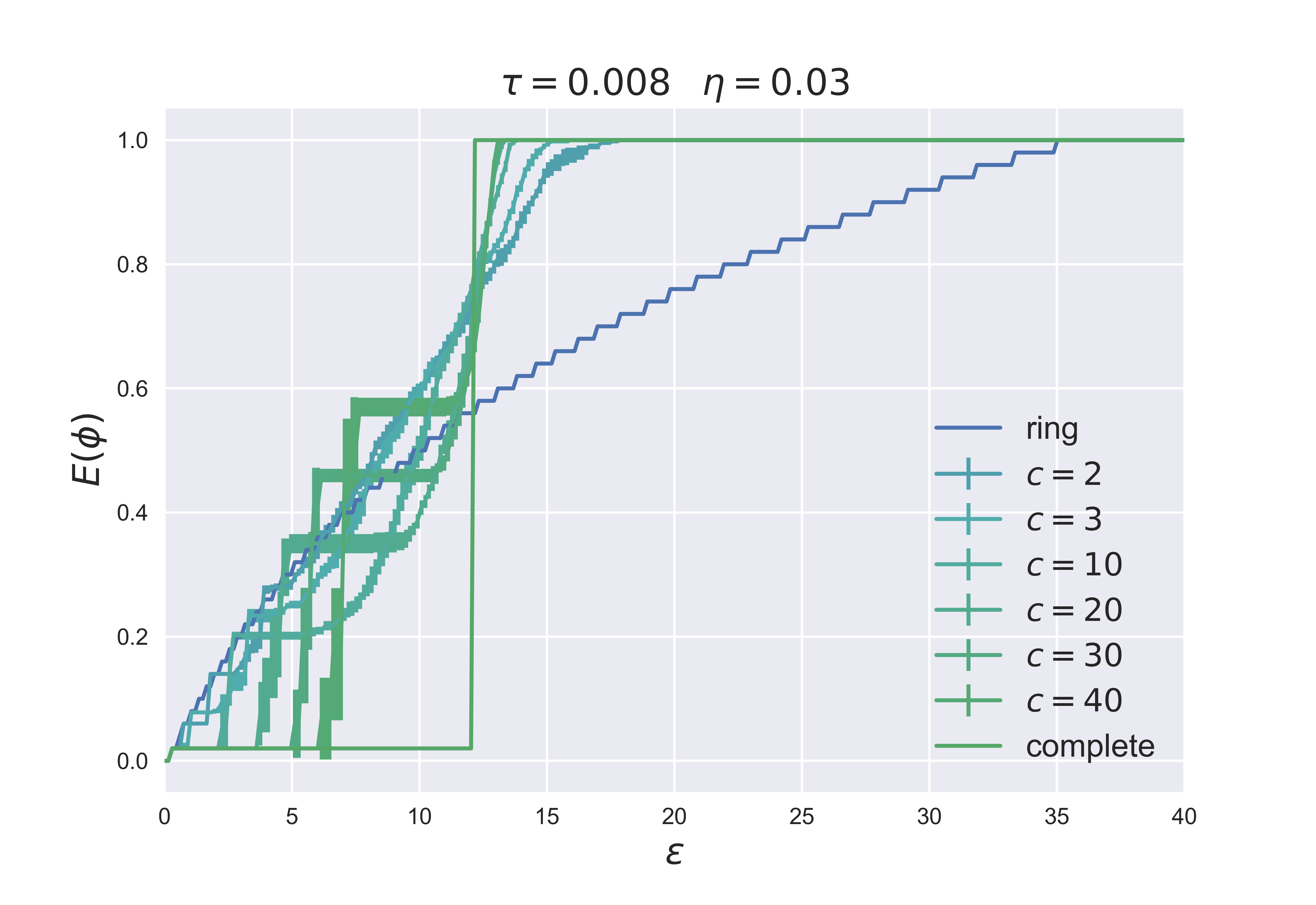}
\includegraphics[width=0.5\linewidth]{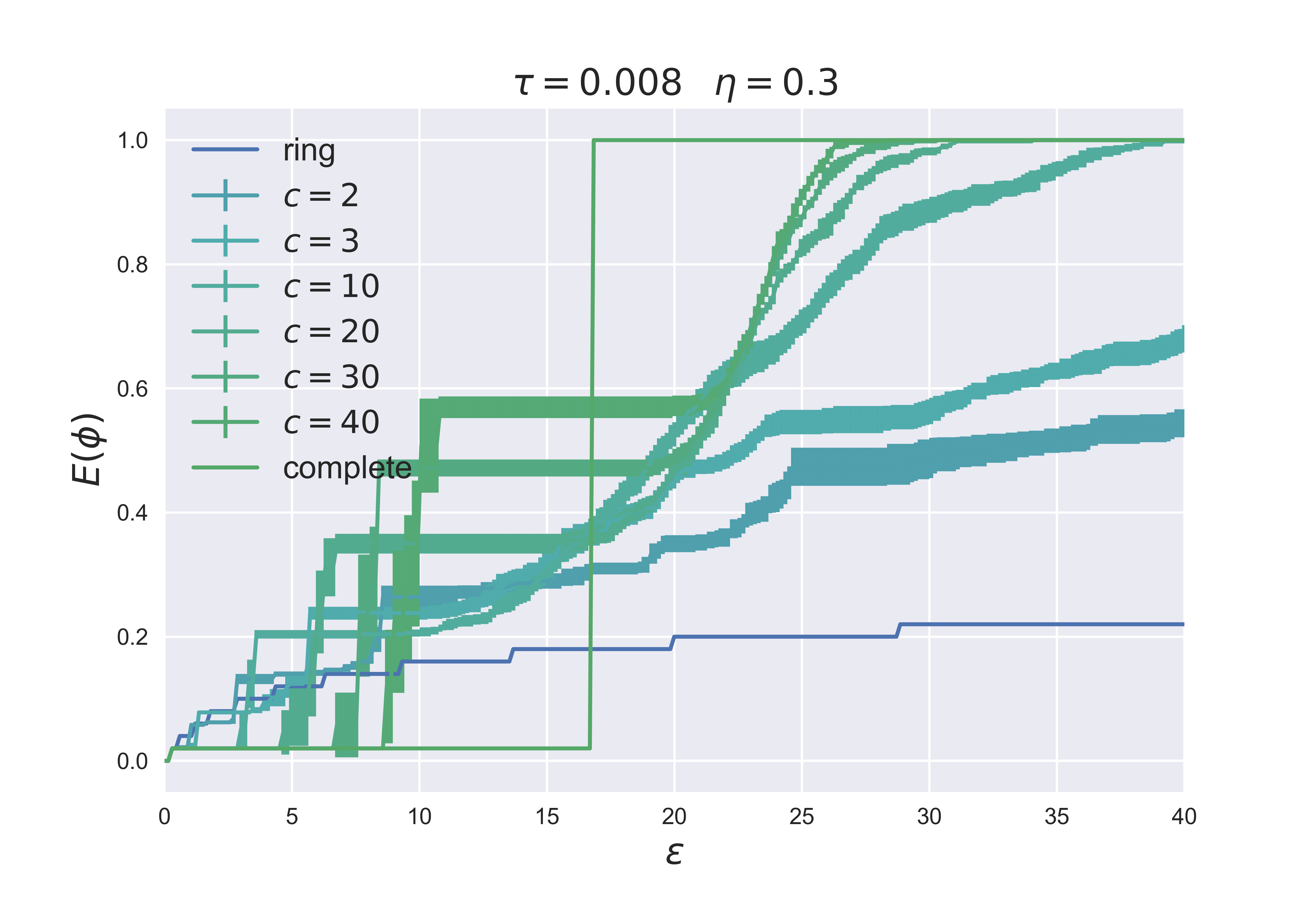}
\caption{Simulation on Random Regular networks with $N=50$, $a=21$, $s=20$, $y=75$.
Results are averaged over $10$ different realizations sampled from a directed configuration model.
\label{fig:CORRd}}
\end{figure}
For $\eta = 0.03$ (left panel) the ring network is more stable than both the complete network and all the intermediate ones.
In fact, while the extent of contagion becomes systemic (=1) in the ring network only for large shocks ($\epsilon \geq 35$), the complete network is the one that first reaches maximal instability (for $\epsilon \geq 12$), followed by the intermediate network.
With respect to the intermediate networks, the lightly interconnected networks are again more stable than the highly interconnected ones, and with a stronger magnitude than the case with no equity conversion.
As shown in the right panel of Fig.~\ref{fig:CORRd}, the stability of the ring and lightly interconnected networks strengthens as $\eta$ increases, but non-linearly. 
For $\eta=0.3$, the ring network reclaims its position as the most stable, never reaching full contagion and achieving maximum contagion at around E$(\boldsymbol{\phi})=20\%$ even for the largest shocks. 
Moreover, and due to the non-linearity of the change, lightly interconnected networks ($c=2$ and $c=3$) are never susceptible to systemic contagion and reach their maximum value at E$(\boldsymbol{\phi}) \approx 60\%$.
Finally, for the set of medium to highly interconnected networks, the extent of contagion grows more than linearly, with the complete network showing a systemic contagion already for small shocks ($\epsilon=17$).
In summary, our results clearly show that the equity conversion is the most beneficial for no or lowly interconnected networks, and is the most detrimental for highly interconnected networks.
In fact, highly interconnected networks are more prone to financial contagion in presence of a CoCos trigger.
Fig.~\ref{fig:critshock} summarizes the results by showing the level of the critical size of the shock $\epsilon^*(\eta, \tau)$ as a function of the conversion value $\eta$. 
\begin{figure}[!htb]
\centering
\includegraphics[width=0.5\linewidth]{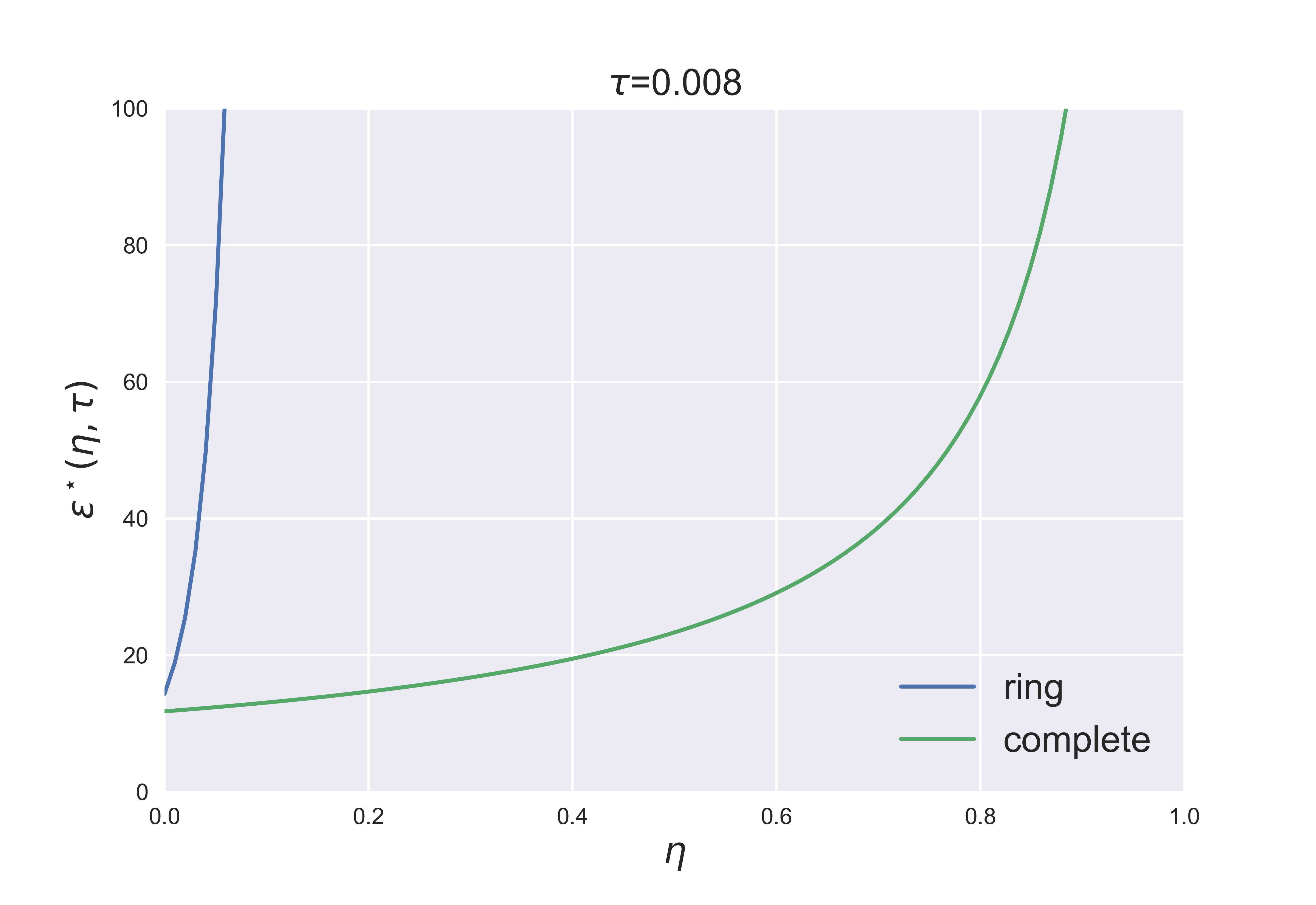}
\caption{Critical shock for the ring and the complete network as a function of $\eta$, see Supplementary Material for the explicit expressions.
}\label{fig:critshock}
\end{figure}
It is evident how, in the ring network, the critical size of the shock grows much faster than in the complete network.

It is important to note that the equity conversion resulting from the trigger might be beneficial for the banking system when the shock amplifies systemic financial risk. 
Conversely, unnecessary conversions can penalize and affect bond (CoCo) holders.
We define the conversions to be unnecessary whenever a shock leads to a systemic conversion without creating systemic risk in the economy.
To better elucidate this result, we compare Fig.~\ref{fig:RR} with Fig.~\ref{fig:critshock}.
Fig.~\ref{fig:RR} shows that in the absence of CoCos the risk is systemic only from $\epsilon \geq 50$, while Fig.~\ref{fig:critshock} illustrates that the ring and the lightly interconnected networks experience a systemic trigger for much smaller shocks in the economy.
From a market perspective, to prevent automatic triggering below a certain threshold, CoCo issuers could consider issuing CoCos with dual-trigger conversion, incorporating both an endogenous trigger linked to the balance sheet or firm's equity value and an exogenous trigger controlled by the regulatory authority (e.g.,~\cite{McDonald_2013}).

\section{Conclusions}
\label{sec: Conclusions}

Contingent convertible bonds (CoCos) are regulatory financial instruments introduced in the aftermath of the 2008-2009 financial crisis with the aim of mitigating systemic risk in the financial system during challenging times. 
In this paper, we present various balance-sheet-based interbank financial networks, both with and without CoCos. 
We demonstrate that the network's structure significantly influences the effectiveness of CoCos as risk-mitigating securities.
Specifically, we demonstrate that for ring and complete networks the state (phase) transition in a network without CoCos is also naturally operational in a network with CoCos, thus confirming the robust-yet-fragile result documented by~\cite{Acemoglou15}.
We also demonstrate that, in the presence of moderate shocks, lightly interconnected networks enhance financial system stability more than highly interconnected ones. Additionally, we highlight the importance of considering the type of interbank financial network for maximizing the effectiveness of CoCos for both issuers and investors. 
Overall, policymakers and regulators should carefully consider the interbank network's role in assessing the potential financial contagion dynamics of CoCos. Despite the focus of the present analysis is on the role of the network connectivity, as a first order characterization of the network topology, it would be crucial for future research to analyze the role of Coco bonds in the case of an interbank network with degree heterogeneity or a more realistic macroscopic structure \cite{van2014finding,bardoscia2017pathways,luu2021collateral,mazzarisi2020dynamic,wilinski2019detectability}.

\section*{Acknowledgements}

This work was partially supported by project SERICS (PE00000014) under the MUR National Recovery and Resilience Plan funded by the European Union - NextGenerationEU. 
DT acknowledges GNFM-Indam for financial support.

\section*{Author contributions statement}
GC was the originator of the project. 
CS and DT  equally contributed to set up the model, conceiving the experiments, analyzing the results and writing the manuscript.

\section*{Data Availability}
The datasets used and/or analysed during the current study are available from the corresponding author on reasonable request.


\bibliography{My_Biblio}

\begin{thebibliography}{10}
\urlstyle{rm}
\expandafter\ifx\csname url\endcsname\relax
  \def\url#1{\texttt{#1}}\fi
\expandafter\ifx\csname urlprefix\endcsname\relax\def\urlprefix{URL }\fi
\expandafter\ifx\csname doiprefix\endcsname\relax\def\doiprefix{DOI: }\fi
\providecommand{\bibinfo}[2]{#2}
\providecommand{\eprint}[2][]{\url{#2}}

\bibitem{battiston2012debtrank}
\bibinfo{author}{Battiston, S.}, \bibinfo{author}{Puliga, M.},
  \bibinfo{author}{Kaushik, R.}, \bibinfo{author}{Tasca, P.} \&
  \bibinfo{author}{Caldarelli, G.}
\newblock \bibinfo{journal}{\bibinfo{title}{Debtrank: Too central to fail?
  financial networks, the fed and systemic risk}}.
\newblock {\emph{\JournalTitle{Scientific reports}}}
  \textbf{\bibinfo{volume}{2}}, \bibinfo{pages}{541} (\bibinfo{year}{2012}).

\bibitem{cimini2015systemic}
\bibinfo{author}{Cimini, G.}, \bibinfo{author}{Squartini, T.},
  \bibinfo{author}{Garlaschelli, D.} \& \bibinfo{author}{Gabrielli, A.}
\newblock \bibinfo{journal}{\bibinfo{title}{Systemic risk analysis on
  reconstructed economic and financial networks}}.
\newblock {\emph{\JournalTitle{Scientific reports}}}
  \textbf{\bibinfo{volume}{5}}, \bibinfo{pages}{15758} (\bibinfo{year}{2015}).

\bibitem{somin2020network}
\bibinfo{author}{Somin, S.}, \bibinfo{author}{Altshuler, Y.},
  \bibinfo{author}{Gordon, G.}, \bibinfo{author}{Pentland, A.} \&
  \bibinfo{author}{Shmueli, E.}
\newblock \bibinfo{journal}{\bibinfo{title}{Network dynamics of a financial
  ecosystem}}.
\newblock {\emph{\JournalTitle{Scientific reports}}}
  \textbf{\bibinfo{volume}{10}}, \bibinfo{pages}{4587} (\bibinfo{year}{2020}).

\bibitem{petrone2018dynamic}
\bibinfo{author}{Petrone, D.} \& \bibinfo{author}{Latora, V.}
\newblock \bibinfo{journal}{\bibinfo{title}{A dynamic approach merging network
  theory and credit risk techniques to assess systemic risk in financial
  networks}}.
\newblock {\emph{\JournalTitle{Scientific Reports}}}
  \textbf{\bibinfo{volume}{8}}, \bibinfo{pages}{5561} (\bibinfo{year}{2018}).

\bibitem{boersma2020reducing}
\bibinfo{author}{Boersma, M.}, \bibinfo{author}{Maliutin, A.},
  \bibinfo{author}{Sourabh, S.}, \bibinfo{author}{Hoogduin, L.} \&
  \bibinfo{author}{Kandhai, D.}
\newblock \bibinfo{journal}{\bibinfo{title}{Reducing the complexity of
  financial networks using network embeddings}}.
\newblock {\emph{\JournalTitle{Scientific Reports}}}
  \textbf{\bibinfo{volume}{10}}, \bibinfo{pages}{17045} (\bibinfo{year}{2020}).

\bibitem{so2022assessing}
\bibinfo{author}{So, M.~K.}, \bibinfo{author}{Mak, A.~S.} \&
  \bibinfo{author}{Chu, A.~M.}
\newblock \bibinfo{journal}{\bibinfo{title}{Assessing systemic risk in
  financial markets using dynamic topic networks}}.
\newblock {\emph{\JournalTitle{Scientific Reports}}}
  \textbf{\bibinfo{volume}{12}}, \bibinfo{pages}{2668} (\bibinfo{year}{2022}).

\bibitem{Merton_1991}
\bibinfo{author}{Merton, R.}
\newblock \bibinfo{journal}{\bibinfo{title}{Distress-contingent convertible
  bonds: A proposed solution to the excess debt problem}}.
\newblock {\emph{\JournalTitle{Harvard Law Review}}}
  \textbf{\bibinfo{volume}{104}}, \bibinfo{pages}{1857--1877}
  (\bibinfo{year}{1991}).

\bibitem{Flannery_2002}
\bibinfo{author}{Flannery, M.~J.}
\newblock \bibinfo{journal}{\bibinfo{title}{No pain, no gain? effecting market
  discipline via ``reverse convertible debentures''}}.
\newblock {\emph{\JournalTitle{Working Paper}}}  (\bibinfo{year}{2002}).

\bibitem{Flannery_2005}
\bibinfo{author}{Flannery, M.~J.}
\newblock \emph{\bibinfo{title}{In Capital Adequacy Beyond Basel: Banking,
  Securities, and Insurance, chapter 5, 171--196}} (\bibinfo{publisher}{Oxford
  University Press}, \bibinfo{year}{2005}).

\bibitem{AVDJIEV2020}
\bibinfo{author}{Avdjiev, S.}, \bibinfo{author}{Bogdanova, B.},
  \bibinfo{author}{Bolton, P.}, \bibinfo{author}{Jiang, W.} \&
  \bibinfo{author}{Kartasheva, A.}
\newblock \bibinfo{journal}{\bibinfo{title}{Co{C}o issuance and bank
  fragility}}.
\newblock {\emph{\JournalTitle{Journal of Financial Economics}}}
  \doiprefix\url{https://doi.org/10.1016/j.jfineco.2020.06.008}
  (\bibinfo{year}{2020}).

\bibitem{Fajardo_2019}
\bibinfo{author}{Fajardo, J.} \& \bibinfo{author}{Mendes, L.}
\newblock \bibinfo{journal}{\bibinfo{title}{On the propensity to issue
  contingent convertible (coco) bonds}}.
\newblock {\emph{\JournalTitle{Quantitative Finance}}}
  \bibinfo{pages}{691--707} (\bibinfo{year}{2019}).

\bibitem{Flannery_2014}
\bibinfo{author}{Flannery, M.~J.}
\newblock \bibinfo{journal}{\bibinfo{title}{Contingent capital instruments for
  large financial institutions: A review of the literature}}.
\newblock {\emph{\JournalTitle{Annual Review of Financial Economics}}}
  \textbf{\bibinfo{volume}{6}}, \bibinfo{pages}{225--240},
  \doiprefix\url{10.1146/annurev-financial-110613-034331}
  (\bibinfo{year}{2014}).

\bibitem{Calomiris_2013}
\bibinfo{author}{Calomiris, C.~W.} \& \bibinfo{author}{Herring, R.~J.}
\newblock \bibinfo{journal}{\bibinfo{title}{How to design a contingent
  convertible debt requirement that helps solve our too-big-to-fail problem}}.
\newblock {\emph{\JournalTitle{Journal of Applied Corporate Finance}}}
  \textbf{\bibinfo{volume}{25}}, \bibinfo{pages}{39--62},
  \doiprefix\url{10.1111/jacf.12015} (\bibinfo{year}{2013}).

\bibitem{Bolton_2012}
\bibinfo{author}{Bolton, P.} \& \bibinfo{author}{Samama, F.}
\newblock \bibinfo{journal}{\bibinfo{title}{Capital access bonds: contingent
  capital with an option to convert}}.
\newblock {\emph{\JournalTitle{Economic Policy}}}
  \textbf{\bibinfo{volume}{27}}, \bibinfo{pages}{275--317},
  \doiprefix\url{10.1111/j.1468-0327.2012.00284.x} (\bibinfo{year}{2012}).

\bibitem{McDonald_2013}
\bibinfo{author}{McDonald, R.~L.}
\newblock \bibinfo{journal}{\bibinfo{title}{Contingent capital with a dual
  price trigger}}.
\newblock {\emph{\JournalTitle{Journal of Financial Stability}}}
  \textbf{\bibinfo{volume}{9}}, \bibinfo{pages}{230 -- 241},
  \doiprefix\url{https://doi.org/10.1016/j.jfs.2011.11.001}
  (\bibinfo{year}{2013}).

\bibitem{Pennacchi_2014}
\bibinfo{author}{Pennacchi, G.}, \bibinfo{author}{Vermaelen, T.} \&
  \bibinfo{author}{Wolff, C. C.~P.}
\newblock \bibinfo{journal}{\bibinfo{title}{Contingent capital: The case of
  coercs}}.
\newblock {\emph{\JournalTitle{Journal of Financial and Quantitative
  Analysis}}} \textbf{\bibinfo{volume}{49}}, \bibinfo{pages}{541--574},
  \doiprefix\url{10.1017/s0022109014000398} (\bibinfo{year}{2014}).

\bibitem{Sundaresan_2015}
\bibinfo{author}{Sundaresan, S.} \& \bibinfo{author}{Wang, Z.}
\newblock \bibinfo{journal}{\bibinfo{title}{On the design of contingent capital
  with a market trigger}}.
\newblock {\emph{\JournalTitle{The Journal of Finance}}}
  \textbf{\bibinfo{volume}{70}}, \bibinfo{pages}{881--920},
  \doiprefix\url{10.1111/jofi.12134} (\bibinfo{year}{2015}).

\bibitem{Glasserman_2016}
\bibinfo{author}{Glasserman, P.} \& \bibinfo{author}{Nouri, B.}
\newblock \bibinfo{journal}{\bibinfo{title}{Market-triggered changes in capital
  structure: Equilibrium price dynamics}}.
\newblock {\emph{\JournalTitle{Econometrica}}} \textbf{\bibinfo{volume}{84}},
  \bibinfo{pages}{2113--2153}, \doiprefix\url{10.3982/ecta11206}
  (\bibinfo{year}{2016}).

\bibitem{Pennacchi_2018}
\bibinfo{author}{Pennacchi, G.} \& \bibinfo{author}{Tchistyi, A.}
\newblock \bibinfo{journal}{\bibinfo{title}{Contingent convertibles with stock
  price triggers: The case of perpetuities}}.
\newblock {\emph{\JournalTitle{The Review of Financial Studies}}}
  \textbf{\bibinfo{volume}{32}}, \bibinfo{pages}{2302--2340},
  \doiprefix\url{10.1093/rfs/hhy092} (\bibinfo{year}{2018}).

\bibitem{Pennacchi_2019}
\bibinfo{author}{Pennacchi, G.} \& \bibinfo{author}{Tchistyi, A.}
\newblock \bibinfo{journal}{\bibinfo{title}{On equilibrium when contingent
  capital has a market trigger: A correction to sundaresan and wang journal of
  finance (2015)}}.
\newblock {\emph{\JournalTitle{The Journal of Finance}}}
  \textbf{\bibinfo{volume}{74}}, \bibinfo{pages}{1559--1576},
  \doiprefix\url{10.1111/jofi.12762} (\bibinfo{year}{2019}).

\bibitem{Hori17}
\bibinfo{author}{Hori, K.} \& \bibinfo{author}{Cer\'{o}n, J.~M.}
\newblock \bibinfo{journal}{\bibinfo{title}{Contingent convertible bonds:\
  payoff structures and incentive effects}}.
\newblock {\emph{\JournalTitle{Birkbeck Working Paper}}}
  \textbf{\bibinfo{volume}{1711}} (\bibinfo{year}{2017}).

\bibitem{Greene_16}
\bibinfo{author}{Greene, R.~W.}
\newblock \bibinfo{journal}{\bibinfo{title}{Understanding {C}o{C}os: What
  operational concerns and global trends mean for {U}.{S}. policymakers}}.
\newblock {\emph{\JournalTitle{M-RCBG Associate Working Paper No. 62}}}
  (\bibinfo{year}{2016}).

\bibitem{Oster_2019}
\bibinfo{author}{Oster, P.}
\newblock \bibinfo{journal}{\bibinfo{title}{Contingent convertible bond
  literature review: making everything and nothing possible?}}
\newblock {\emph{\JournalTitle{Journal of Banking Regulation}}}
  \doiprefix\url{10.1057/s41261-019-00122-z} (\bibinfo{year}{2019}).

\bibitem{Kashyap_2008}
\bibinfo{author}{Kashyap, A.~K.}, \bibinfo{author}{Rajan, R.~G.} \&
  \bibinfo{author}{Stein, J.~C.}
\newblock \bibinfo{title}{Maintaining stability in a changing financial
  system}.
\newblock In \emph{\bibinfo{booktitle}{Rethinking Capital Regulation, Economic
  Symposium}} (\bibinfo{year}{2008}).

\bibitem{Squam_2010}
\bibinfo{author}{{The Squam Lake Report}}.
\newblock \bibinfo{journal}{\bibinfo{title}{An expedited resolution mechanism
  for distressed financial firms: regulatory hybrid securities, the squam lake
  report fixing the financial system}}.
\newblock {\emph{\JournalTitle{Working Group on Financial Regulation}}}
  (\bibinfo{year}{2010}).

\bibitem{Koziol_2012}
\bibinfo{author}{Koziol, C.} \& \bibinfo{author}{Lawrenz, J.}
\newblock \bibinfo{journal}{\bibinfo{title}{Contingent convertibles. solving or
  seeding the next banking crisis?}}
\newblock {\emph{\JournalTitle{Journal of Banking \& Finance}}}
  \textbf{\bibinfo{volume}{36}}, \bibinfo{pages}{90--104},
  \doiprefix\url{10.1016/j.jbankfin.2011.06.009} (\bibinfo{year}{2012}).

\bibitem{Hilscher14}
\bibinfo{author}{Hilscher, J.} \& \bibinfo{author}{Raviv, A.}
\newblock \bibinfo{journal}{\bibinfo{title}{Bank stability and market
  discipline: The effect of contingent capital on risk taking and default
  probability}}.
\newblock {\emph{\JournalTitle{Journal of Corporate Finance}}}
  \textbf{\bibinfo{volume}{29}}, \bibinfo{pages}{542--560}
  (\bibinfo{year}{2014}).

\bibitem{Berg11}
\bibinfo{author}{Berg, T.} \& \bibinfo{author}{Kaserer, C.}
\newblock \bibinfo{journal}{\bibinfo{title}{Does contingent capital induce
  excessive risk-taking and prevent an efficient recapitalization of banks?}}
\newblock {\emph{\JournalTitle{Systemic Risk, Basel III, Financial Stability
  and Regulation 2011}}}  (\bibinfo{year}{2011}).

\bibitem{Chan_2017}
\bibinfo{author}{Chan, S.} \& \bibinfo{author}{Wijnbergen, S.~V.}
\newblock \bibinfo{journal}{\bibinfo{title}{Coco design, risk shifting and
  financial fragility}}.
\newblock {\emph{\JournalTitle{ECMI Working Paper No. 2}}}
  (\bibinfo{year}{2017}).

\bibitem{Goncharenko_2017}
\bibinfo{author}{Goncharenko, R.}, \bibinfo{author}{Ongena, S.} \&
  \bibinfo{author}{Rauf, A.}
\newblock \bibinfo{journal}{\bibinfo{title}{The agency of {C}o{C}o: Why do
  banks issue contingent convertible bonds?}}
\newblock {\emph{\JournalTitle{SSRN Electronic Journal}}}
  \doiprefix\url{10.2139/ssrn.3067909} (\bibinfo{year}{2017}).

\bibitem{Martynova_2018}
\bibinfo{author}{Martynova, N.} \& \bibinfo{author}{Perotti, E.}
\newblock \bibinfo{journal}{\bibinfo{title}{Convertible bonds and bank
  risk-taking}}.
\newblock {\emph{\JournalTitle{Journal of Financial Intermediation}}}
  \textbf{\bibinfo{volume}{35}}, \bibinfo{pages}{61--80},
  \doiprefix\url{10.1016/j.jfi.2018.01.002} (\bibinfo{year}{2018}).

\bibitem{Albul_2013}
\bibinfo{author}{Albul, B.}, \bibinfo{author}{Jaffee, D.~M.} \&
  \bibinfo{author}{Tchistyi, A.}
\newblock \bibinfo{journal}{\bibinfo{title}{Contingent convertible bonds and
  capital structure decisions}}.
\newblock {\emph{\JournalTitle{Working Paper}}}  (\bibinfo{year}{2013}).

\bibitem{Chen_2017_RFS}
\bibinfo{author}{Chen, N.}, \bibinfo{author}{Glasserman, P.},
  \bibinfo{author}{Nouri, B.} \& \bibinfo{author}{Pelger, M.}
\newblock \bibinfo{journal}{\bibinfo{title}{Contingent capital, tail risk, and
  debt-induced collapse}}.
\newblock {\emph{\JournalTitle{The Review of Financial Studies}}}
  \textbf{\bibinfo{volume}{30}}, \bibinfo{pages}{3921--3969},
  \doiprefix\url{10.1093/rfs/hhx067} (\bibinfo{year}{2017}).

\bibitem{Goncharenko_2019}
\bibinfo{author}{Goncharenko, R.}
\newblock \bibinfo{journal}{\bibinfo{title}{Fighting fire with gasoline:
  {C}o{C}os in lieu of equity}}.
\newblock {\emph{\JournalTitle{SSRN Working Paper}}}  (\bibinfo{year}{2019}).

\bibitem{Gupta2020}
\bibinfo{author}{Gupta, A.}, \bibinfo{author}{Wang, R.} \& \bibinfo{author}{Lu,
  Y.}
\newblock \bibinfo{journal}{\bibinfo{title}{Addressing systemic risk using
  contingent convertible debt - a network analysis}}.
\newblock {\emph{\JournalTitle{European Journal of Operational Research}}}
  \doiprefix\url{https://doi.org/10.1016/j.ejor.2020.07.062}
  (\bibinfo{year}{2020}).

\bibitem{Haldane_14}
\bibinfo{author}{Haldane, A.~G.}
\newblock \bibinfo{title}{Managing global finance as a system}.
\newblock In \emph{\bibinfo{booktitle}{Maxwell Fry Global Finance Lecture}}
  (\bibinfo{year}{2014}).

\bibitem{Acemoglou15}
\bibinfo{author}{Acemoglu, D.}, \bibinfo{author}{Ozdaglar, A.} \&
  \bibinfo{author}{Tahbaz-Salehi, A.}
\newblock \bibinfo{journal}{\bibinfo{title}{Systemic risk and stability in
  financial networks}}.
\newblock {\emph{\JournalTitle{American Economic Review}}}
  \textbf{\bibinfo{volume}{105}}, \bibinfo{pages}{564--608}
  (\bibinfo{year}{2015}).

\bibitem{Gai_2010}
\bibinfo{author}{Gai, P.} \& \bibinfo{author}{Kapadia, S.}
\newblock \bibinfo{journal}{\bibinfo{title}{Contagion in financial networks}}.
\newblock {\emph{\JournalTitle{Proceedings of the Royal Society A:
  Mathematical, Physical and Engineering Sciences}}}
  \textbf{\bibinfo{volume}{466}}, \bibinfo{pages}{2401--2423},
  \doiprefix\url{10.1098/rspa.2009.0410} (\bibinfo{year}{2010}).

\bibitem{brini2023reinforcement}
\bibinfo{author}{Brini, A.}, \bibinfo{author}{Tedeschi, G.} \&
  \bibinfo{author}{Tantari, D.}
\newblock \bibinfo{journal}{\bibinfo{title}{Reinforcement learning policy
  recommendation for interbank network stability}}.
\newblock {\emph{\JournalTitle{Journal of Financial Stability}}}
  \bibinfo{pages}{101139} (\bibinfo{year}{2023}).

\bibitem{Allen_2000}
\bibinfo{author}{Allen, F.} \& \bibinfo{author}{Gale, D.}
\newblock \bibinfo{journal}{\bibinfo{title}{Financial contagion}}.
\newblock {\emph{\JournalTitle{Journal of Political Economy}}}
  \textbf{\bibinfo{volume}{108}}, \bibinfo{pages}{1--33},
  \doiprefix\url{10.1086/262109} (\bibinfo{year}{2000}).

\bibitem{Freixas_2000}
\bibinfo{author}{Freixas, X.}, \bibinfo{author}{Parigi, B.~M.} \&
  \bibinfo{author}{Rochet, J.-C.}
\newblock \bibinfo{journal}{\bibinfo{title}{Systemic risk, interbank relations,
  and liquidity provision by the central bank}}.
\newblock {\emph{\JournalTitle{Journal of Money, Credit and Banking}}}
  \textbf{\bibinfo{volume}{32}}, \bibinfo{pages}{611--638}
  (\bibinfo{year}{2000}).

\bibitem{bollobas1980probabilistic}
\bibinfo{author}{Bollob{\'a}s, B.}
\newblock \bibinfo{journal}{\bibinfo{title}{A probabilistic proof of an
  asymptotic formula for the number of labelled regular graphs}}.
\newblock {\emph{\JournalTitle{European Journal of Combinatorics}}}
  \textbf{\bibinfo{volume}{1}}, \bibinfo{pages}{311--316}
  (\bibinfo{year}{1980}).

\bibitem{Newman}
\bibinfo{author}{Newman, M.}
\newblock \bibinfo{journal}{\bibinfo{title}{Networks: An introduction}}.
\newblock {\emph{\JournalTitle{Oxford University Press}}}
  (\bibinfo{year}{2010}).

\bibitem{van2014finding}
\bibinfo{author}{Van~Lelyveld, I.} \emph{et~al.}
\newblock \bibinfo{journal}{\bibinfo{title}{Finding the core: Network structure
  in interbank markets}}.
\newblock {\emph{\JournalTitle{Journal of Banking \& Finance}}}
  \textbf{\bibinfo{volume}{49}}, \bibinfo{pages}{27--40}
  (\bibinfo{year}{2014}).

\bibitem{bardoscia2017pathways}
\bibinfo{author}{Bardoscia, M.}, \bibinfo{author}{Battiston, S.},
  \bibinfo{author}{Caccioli, F.} \& \bibinfo{author}{Caldarelli, G.}
\newblock \bibinfo{journal}{\bibinfo{title}{Pathways towards instability in
  financial networks}}.
\newblock {\emph{\JournalTitle{Nature communications}}}
  \textbf{\bibinfo{volume}{8}}, \bibinfo{pages}{14416} (\bibinfo{year}{2017}).

\bibitem{luu2021collateral}
\bibinfo{author}{Luu, D.~T.}, \bibinfo{author}{Napoletano, M.},
  \bibinfo{author}{Barucca, P.} \& \bibinfo{author}{Battiston, S.}
\newblock \bibinfo{journal}{\bibinfo{title}{Collateral unchained:
  Rehypothecation networks, concentration and systemic effects}}.
\newblock {\emph{\JournalTitle{Journal of Financial Stability}}}
  \textbf{\bibinfo{volume}{52}}, \bibinfo{pages}{100811}
  (\bibinfo{year}{2021}).

\bibitem{mazzarisi2020dynamic}
\bibinfo{author}{Mazzarisi, P.}, \bibinfo{author}{Barucca, P.},
  \bibinfo{author}{Lillo, F.} \& \bibinfo{author}{Tantari, D.}
\newblock \bibinfo{journal}{\bibinfo{title}{A dynamic network model with
  persistent links and node-specific latent variables, with an application to
  the interbank market}}.
\newblock {\emph{\JournalTitle{European Journal of Operational Research}}}
  \textbf{\bibinfo{volume}{281}}, \bibinfo{pages}{50--65}
  (\bibinfo{year}{2020}).

\bibitem{wilinski2019detectability}
\bibinfo{author}{Wilinski, M.}, \bibinfo{author}{Mazzarisi, P.},
  \bibinfo{author}{Tantari, D.} \& \bibinfo{author}{Lillo, F.}
\newblock \bibinfo{journal}{\bibinfo{title}{Detectability of macroscopic
  structures in directed asymmetric stochastic block model}}.
\newblock {\emph{\JournalTitle{Physical Review E}}}
  \textbf{\bibinfo{volume}{99}}, \bibinfo{pages}{042310}
  (\bibinfo{year}{2019}).

\end{thebibliography}

\end{document}